\documentclass{article}
\usepackage{amsmath}
\usepackage{amsthm}
\usepackage{enumitem}
\usepackage{graphicx}
\usepackage{subcaption}
\usepackage{xcolor}
\usepackage{amsfonts} 
\usepackage{verbatim}
\usepackage{xfrac}

\usepackage{natbib}

\newtheorem{lemma}{Lemma}
\newtheorem{theorem}{Theorem}
\newtheorem{definition}{Definition}
\newtheorem{proposition}{Proposition}

\DeclareMathOperator{\interior}{int}
\newcommand{\mc}[1]{\mathcal{#1}} % Shortcut for mathcal
\newcommand{\eqn}[1]{\eqref{eqn:#1}} % Refer to an equation
\newcommand{\alttime}{u}

\title{Departure Time Choice with Parametric Heterogeneity: Equilibrium and Instability}
\author{Hillel Bar-Gera\footnote{Ben Gurion University},  Stephen D. Boyles\footnote{The University of Texas at Austin},  Liron Ravner\footnote{University of Haifa}}

\date{\today}

\begin{document}

\maketitle

\begin{abstract} 
Vickrey's classic single-bottleneck departure time choice equilibrium model exhibits instability under many plausible day-to-day learning dynamics. Such instability is not observed in reality --- does this difference stem from the day-to-day dynamics or from one of the simplifying assumptions of the basic model?  This paper explores a variant of the basic model with a continuous distribution of schedule delay parameters which we intuitively expect to have more favorable stability properties.  To attain tractability we assume a monotonic relationship between earliness and lateness parameters. 
We first verify the existence and uniqueness of the equilibrium solution for this model. We then study a broad class of day-to-day dynamics satisfying local pressure and order preservation conditions. Our main contribution is a formal proof that, surprisingly, all such day-to-day dynamics in this context are unstable. 
\end{abstract}

\section{ INTRODUCTION}

Vickrey's work on departure time choice inspired a substantial number of subsequent studies \citep{vickrey69}. A relatively recent bibliometric review identified 232 journal publications in 50 years that relied on Vickrey's work~\citep{li20}.
These subsequent studies have proceeded along many different lines.  Among these are extensions of Vickrey's model in which travelers are heterogeneous in some way (e.g., schedule delay parameters), and investigations of the stability properties of the departure time choice equilibrium under perturbations.  Stability is important both for practical and algorithmic reasons.  Practically, stable equilibria are more likely to be observed in the field.  Algorithmically, stability is helpful when designing iterative algorithms that converge to an equilibrium solution.  While the equilibrium in Vickrey's original model and in some extensions can be derived in closed form, numerical computations are important when embedding departure time choice in larger network models, or in other more complex extensions.  As described in more detail in the literature review, most stability results to date are negative, and indicate that the original Vickrey bottleneck is unstable under many natural dynamic processes.

These two streams of research (heterogeneity and stability) have proceeded more or less independently of each other, but we believe there is good reason to consider them together.  In many cases, traveler heterogeneity implies a unique departure order at equilibrium, whereas in the homogeneous model the travelers are fully interchangeable.  Furthermore, in the standard model, at equilibrium each traveler is indifferent among all possible departure times during the peak, whereas heterogeneous models can be formulated where each traveler's departure time has \emph{strictly} minimum cost.  Both of these properties intuitively suggest that stability is easier to achieve in a heterogeneous setting.
The main contribution of our paper is to investigate this intuition rigorously.
Surprisingly, we find that it is not true, and that equilibrium is unstable in such settings.

More specifically, we explore a variant of the basic model with a continuous distribution of schedule delay parameters, and investigate stability properties of the equilibrium solution.  
We model traveler heterogeneity with a monotonic relationship between the earliness and lateness parameters.
As a result, the general case where the parameter space is two-dimensional is replaced with a single-dimensional parameter space.
We believe such a model is plausible behaviorally, and considerably simpler to analyze than a general joint distribution of schedule delay parameters.
We prove equilibrium existence and uniqueness under strict monotonicity, and show that the travelers' order at equilibrium is known \emph{a priori}. 

We then consider a general family of day-to-day dynamics, or iterative processes.  Our analysis covers any iterative process that satisfies two assumptions: i) the \emph{a priori} order is maintained throughout the iterative process; and ii) each traveler acts myopically to improve her own conditions.
A formal definition of these assumptions, as well as the motivations for adopting them, will be discussed in Section \ref{Sec:InstabilityDiscrete}.

Our main finding is that any such iterative process cannot guarantee convergence to equilibrium, even within arbitrarily small neighborhoods around the equilibrium. In that sense, the equilibrium solution is fundamentally unstable, despite having a unique departure order, and despite each traveler departing at a time which is uniquely minimal for her.  Our main finding is thus also negative, and suggests that the cause of instability in departure time choice models is more fundamental.  Implications for other variations of departure time choice models, including relaxing the monotonicity assumption, are presented in the discussion section at the end of the paper.

\section{Literature review}
\label{sec:litreview}

In Vickrey's original model, all travelers are identical in their value of time, and in their scheduling preferences.
Since then, researchers have generalized Vickrey's model to consider different kinds of heterogeneity 
among travelers.
This can include variations in the desired arrival time, as in \cite{hendrickson81}, \cite{smith84}, and \cite{daganzo85}, or in the penalty factors for queueing delay, early arrival, and late arrival, as in \cite{newell87},  \cite{arnott88}, and \cite{arnott94}.
Both continuous and discrete forms of heterogeneity have been studied; see Section 3.1.6 and Table 6 in \cite{li20} for a comprehensive review of traveler heterogeneity in the single bottleneck model.

In many heterogeneous formulations, travelers obey an ordering principle at equilibrium. 
An interesting framework for demonstrating travelers' order in heterogeneous departure time choice models was proposed by \cite{aka21}, combining results of several previous studies on this topic. Their analysis relied primarily on the equivalence between the system optimum solution for departure time choice and the linear formulation of the `total transport' problem, with a continuous time dimension. They discussed sufficient conditions on the schedule cost derivative for interpreting Lagrange multipliers as queue length, thus demonstrating the equivalence between the system optimum and user equilibrium solutions. 
They also showed that if schedule costs satisfy the Monge condition, for a certain permutation of traveler classes, then arrivals at equilibrium must follow a corresponding `sorting pattern'. 
This is true in one-dimensional heterogeneity such as: desired arrival time; earliness penalty (without late arrivals); and lateness penalty (without early arrivals). In each of these cases additional assumptions were made (e.g. convex or quasi-convex schedule cost). In the case of two-dimensional heterogeneity with respect to both earliness and lateness penalties, a sorting pattern occurs among early arrivals and among late arrivals, due to properties of two subproblems of the main system optimal problem.

Although many researchers have studied heterogeneity in different forms, comparatively few have addressed the issue of solution stability. Only two of the studies listed in the review of \cite{li20} focus on the day-to-day dynamics \citep{gou18b, zhu19}, and both assumed bounded rationality.
To our knowledge, the first to consider the stability of day-to-day learning dynamics in the context of departure time choice, and particularly in the context of Vickrey's bottleneck model, was \cite{depalma00}.
In his paper, de~Palma proposed four dynamic evolution strategies, aiming to shift travelers towards lower-cost departure times, and evaluates their performance numerically.
His experiments showed that these dynamics often lead to oscillatory behavior that fails to converge; one exception is a logit-based dynamic which was seen to converge to a stochastic user equilibrium when unobserved costs are sufficiently high relative to queueing and schedule delay (a paradoxical finding, since convergence to equilibrium seems easiest when travelers care about it the least).
\cite{iryo08} provided numerical evaluation of the `proportional-switch' dynamics. Subsequently, \cite{iryo19} proved that a `continuous-day' version of the replicator dynamics is unstable. \cite{gou18} extended these results and proved the lack of stability for the `continuous-day' versions of additional dynamics: proportional-switch, the \emph{t\^{a}tonnement} process, the simplex gravity, the projected dynamical system, and the evolutionary traffic dynamics. They also provided extensive numerical examples illustrating the lack of stability in iterative (discrete-day) versions of these dynamics. 

The instability results in the departure time choice literature stand in contrast to classical game theoretic results on convergence of rational evolutionary dynamics, as discussed in detail by \cite{sandholm2010}. Notably, best response dynamics are known to converge to equilibrium in potential games \citep{monderer1996} and congestion games \citep{monderer1996}. 

If travelers' rationality is bounded, convergence may emerge \citep{gou18b, zhu19}, but to solutions that may deviate considerably from the user equilibrium \citep[Fig.~1][]{daly21}. 
\cite{lamotte21} studied the underlying causes for instability, deriving conditions under which the cost function mapping is monotone, as it is easier to develop stable dynamics in such cases. \cite{daly21} explored unique dynamics that exhibited relatively stable behavior. However, the dynamics studied by \cite{daly21} require a large number of iterations to converge. Moreover, the analysis lacks a proof of convergence and presents only numerical results.

\cite{jin21} studied stability of dynamics for the homogeneous user model, in the continuous limit based on a discrete model of switching departure times among time slots.  These dynamics are local, in the sense that the departure rate in each interval changes only if there is cost improvement in the adjacent time slots.  Convergence to equilibrium is seen in his numerical experiments.  Generalizing these dynamics to other variants of the basic model is nontrivial, as several parameters must be chosen to satisfy several conditions.

A number of empirical studies have shown that traffic patterns are relatively stable from day to day, after controlling for factors such as day of week, season, and the presence of incidents~\citep{rakha95,weijermars05,wang15,crawford17}.  In this light, the challenge of finding stable dynamic processes for departure time choice equilibrium is especially puzzling.  Either finding behaviorally-plausible learning dynamics that lead to equilibrium, or understanding what specific features of bottleneck or behavior models lead to instability in learning dynamics, would be very valuable in connecting theoretical traffic models to observed evidence.

To summarize, the few studies that examined the issue raise doubts about the stability of Vickrey's original departure time choice model under plausible day-to-day learning dynamics. Such instability is not observed in reality. The key question is whether the difference between reality and model results stems from some specific features of the day-to-day dynamics that have been studied, or from one of the simplifying assumptions of the basic model.

In particular, we think it plausible that traveler homogeneity may contribute to instability.
As discussed in the papers above, in heterogeneous flavors of Vickrey's model there is often an ordering property which must hold at equilibrium, and each traveler departs at a time of \emph{strictly} minimal cost for herself.
The homogeneous model has no such natural ordering, and at equilibrium each traveler is indifferent among any departure time during the peak period.
These observations suggest that equilibrium in the heterogeneous case may be more stable to perturbations than in the homogeneous case, motivating the model we present below.

\section{Model and notation}
\label{sec:model}

We model a single bottleneck that discharges at capacity $s$ whenever vehicles are waiting in queue.
A total of $N$ travelers must select a departure time in order to minimize their generalized cost, which is a function of queueing delay and arrival time relative to a preferred time.
We measure time relative to this preferred arrival time, so positive time values correspond to late arrival, and negative time values to early arrival.
Many models of this type have been studied.
We adopt an individual cost function of the typical form
\[
C(t) = Q(t) + \beta [-t]^+ + \gamma[t]^+,
\,\]
where $C(t)$ is the generalized cost for a traveler arriving at time $t$; $Q(t)$ is the queueing delay experienced by a traveler arriving at $t$; $\beta$ and $\gamma$ are the penalty parameters, that is relative per-minute costs of early and late arrival, compared to the per-minute cost of queueing delay; and $[\cdot]^+ = \max\{\cdot, 0\}$ is the positive part of its argument.

We consider a case where the values of $\beta$ and $\gamma$ vary over the population, but in a way that can be parameterized by a single variable, rather than the more general case of a two-dimensional distribution.
The earliness penalty distribution is specified by the cumulative distribution function (CDF) $F_\beta$, or equivalently by a scaled inverse function, $\beta(n) = F_\beta^{-1}(n/N)$.
By definition the function $\beta(n)$ is monotone increasing.
We further assume that it is continuous, implying that the support of $\beta$ is a single real interval.
For reasons explained below, we will require that this support be contained within the interval $(0, 1)$, that is, that queueing delay is always strictly worse than arriving early by an equivalent amount.

Similarly, the lateness penalty distribution is given by $F_\gamma$, or equivalently by the survival function (complementary CDF) $\tilde{F}_\gamma = 1 - F_\gamma$.
The distribution can also be represented by a scaled inverse function $\gamma(n) = \tilde{F}_\gamma^{-1}(n/N)$, which is a monotone decreasing function.
As with $\beta(n)$, we assume that $\gamma(n)$ is continuous, and therefore has connected support.
The support of $\gamma$ is contained in $(0, \infty)$.

\begin{figure}
\begin{subfigure}{0.45\textwidth}
  \centering
    \includegraphics[scale=0.7]{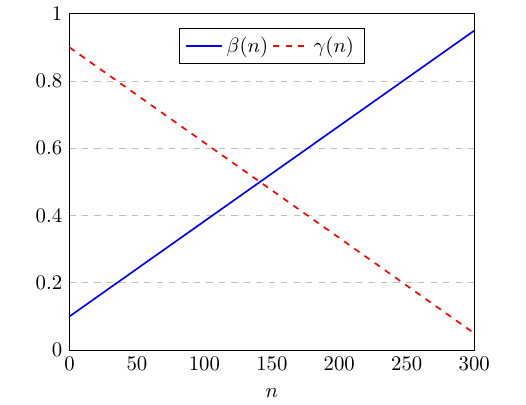}\label{fig:beta_gamma_linear}
\end{subfigure}
\begin{subfigure}{0.45\textwidth}
  \centering
\includegraphics[scale=0.7]{ 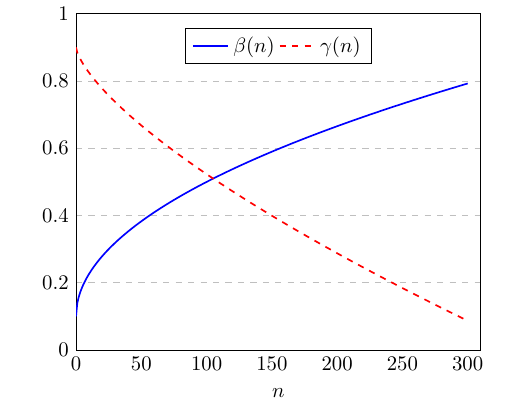}\label{fig:beta_gamma_nonlinear}
\end{subfigure}
\caption{The figure illustrates two examples of penalty functions (linear on the left and non-linear on the right). For every type $n\in[0,300]$, the solid blue line represents the earliness penalty and the dashed red line represents the lateness penalty. A higher index $n$ indicates a higher earliness penalty and lower lateness penalty, and vice versa. }\label{fig:penaltyexample}
\end{figure}

Our key assumption is that the same continuous index $n \in [0,N]$ can be used to represent both a traveler's early and late arrival penalty values. (See Figure~\ref{fig:penaltyexample}.)

That is, $\beta(n)$ and $\gamma(n)$ both reflect the same infinitesimal traffic fraction.
Effectively, we assume that ordering travelers by increasing value of $\beta$ is the same as ordering them by decreasing value of $\gamma$.
We discuss the implications of this assumption at greater length at the end of this section.
We will show below that this ordering --- from travelers with lowest $\beta$ (and highest $\gamma$) to travelers with the lowest $\gamma$ (and highest $\beta$) --- must also be the order in which travelers arrive in an equilibrium solution, but will first define our model more generally.

Each traveler chooses a time at which to depart.
However, as we show below, we can formulate the equilibrium and model in terms of \emph{arrival} times instead, which facilitates explanation and analysis.
We will use the notation $t(n)$ to refer to the time at which the traveler of index $n$ arrives, and $n(t)$ to refer to the index of a traveler arriving at time $t$.
These functions are defined as inverses of each other, so specifying one is equivalent to specifying the other.

The generalized cost function for traveler $n$ arriving at time $t$ assuming queuing profile $Q(t)$ is given by
\begin{equation}
\label{eqn:costfunction}
C_n(t) = Q(t) + \beta(n) [-t]^+ + \gamma(n) [t]^+ 
\,.
\end{equation}

We will represent equilibrium in terms of the queue evolution $Q(t)$, and the traveler-to-arrival time mapping $t(n)$.
It is more common to represent equilibrium in terms of the departure time choice, rather than in these terms.
As we show in Appendix~\ref{sec:choicerepresentation}, defining equilibrium in terms of departure time, or in terms of $Q(t)$ and $t(n)$, is equivalent, and the latter representation is more convenient for analysis.
We denote a queueing delay profile by $Q:=\{Q(t): t\in\mathcal{T}\}$, and assume that it is defined on the interval $\mathcal{T} = [-N/s, N/s]$, and its support must have a total width no greater than $N/s$. 

These conditions enforce that the bottleneck discharges at capacity whenever the queue is active.
The queue obeys first-in, first-out discipline, which requires that if $t_1 < t_2$, we must have $Q(t_2) - Q(t_1) < t_2 - t_1$.
Positive jumps are therefore not allowed, but negative jumps are possible and correspond to a gap in departures while the queue is active.
This condition also enforces $Q'(t) < 1$ wherever $Q$ is differentiable.
Under these conditions, we define equilibrium as follows.

\begin{definition}
\label{def:equilibrium}
The functions $Q(t)$ and $t(n)$ are an equilibrium if
\begin{equation}
\label{eqn:equilibrium}
C_n(t(n)) = \min_{\alttime \in \mathcal{T}} \{ C_n(\alttime) \},
\end{equation}
for all $n \in [0, N]$,
with $C_n$ given in terms of $Q$ by equation~\eqn{costfunction}.
\end{definition}
The classic Vickrey bottleneck is a limiting case of our model, by taking the support of $\beta$ and $\gamma$ to be a single point, i.e., $\beta(n)=\beta$ and $\gamma(n)=\gamma$ for all $n$.

Figure~\ref{fig:cost} illustrates the individual optimality condition for the different traveler types. 
Each plotted line represents the cost function $C_n(t)$ for a specific traveler $n$, with penalties $\beta(n)$ and $\gamma(n)$, as a function of all possible arrival times. At equilibrium, the actual arrival time $t(n)$ attains the global minimum for the contour of every traveler.
The example used to generate this figure assumes linear earliness and lateness penalty functions, as demonstrated next.

\begin{figure}
  \centering
    \includegraphics[scale=0.7]{ 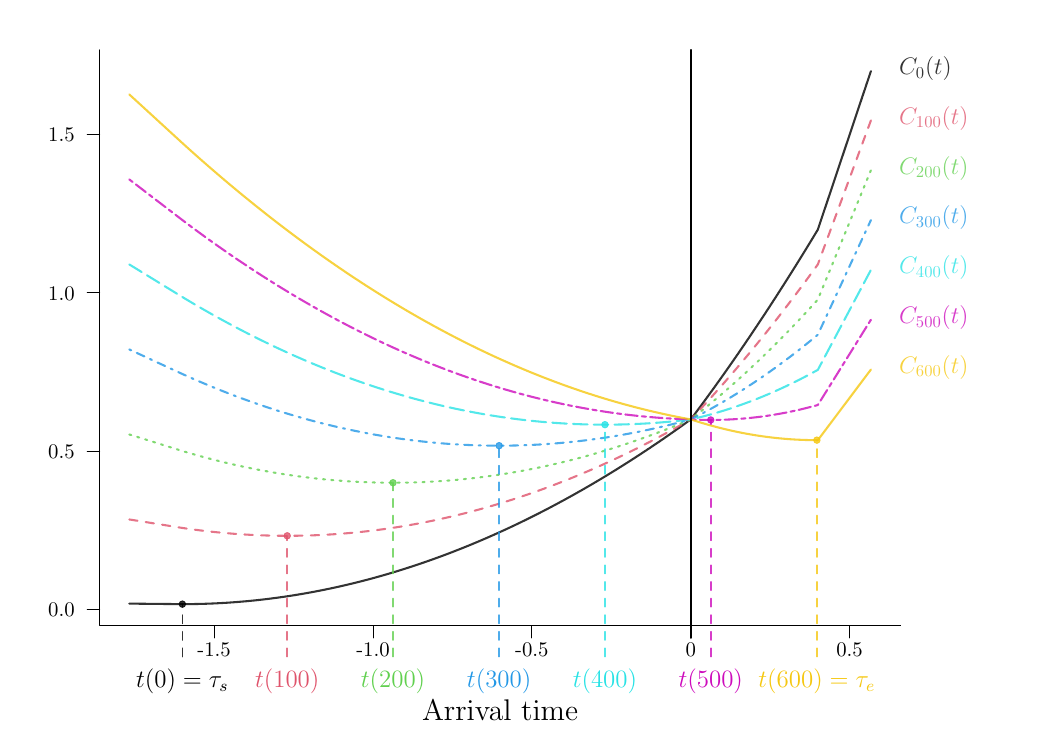}
    \caption{The cost functions $C_n(t)$ are plotted for different traveler types. In equilibrium, the minimal cost for every type $n$ traveler is attained at $t(n)$.}\label{fig:cost}
\end{figure}

\begin{figure}
  \centering
    \includegraphics[width=11cm, height=7cm]{ 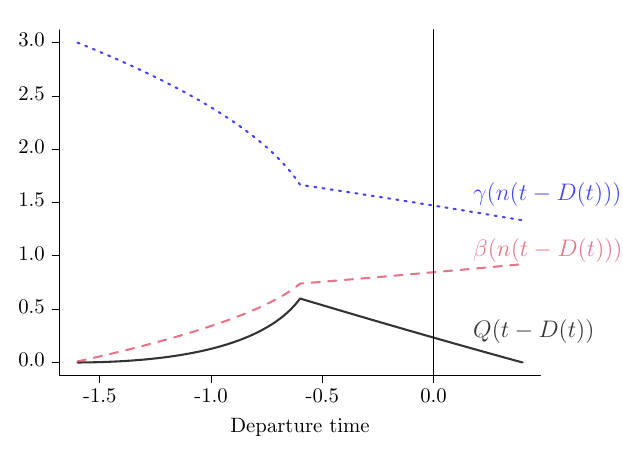}
    \caption{The equilibrium queue, earliness and lateness penalties of travelers as a function of their departure time.}\label{fig:dep_queue}
\end{figure}

As a specific demonstration of our model, assume that there are $N = 600$ travelers, and the bottleneck capacity is $s = 300$.
The early and late penalty factors are distributed uniformly between $[0.01, 0.91]$ and $[1.2, 3]$ respectively, so
\begin{align*}
\beta(n) &= 0.01 + 0.0015n \\
\gamma(n) &= 3 - 0.003n \,.
\end{align*}

In this case, for a given initial departure time $\tau_s$, the queueing profile
\[
Q(t) =
\begin{cases}
0.01 (t - \tau_s) + 0.225 (t - \tau_s)^2 & t \in [\tau_s, 0] \\
(0.225 \tau_s^2 - 0.01 \tau_s) + 0.45 t^2 -  (3 + 0.9 \tau_s)t  & t \in [0, \tau_e]
\end{cases}
\,,
\]
is a user equilibrium if travelers arrive in order of increasing $\beta$ (equivalently, decreasing $\gamma$), with $\tau_e = \tau_s + N/s = \tau_s + 2$, and with $\tau_s$ being the solution to the quadratic equation $Q(\tau_e) = Q(\tau_s + 2) = 0$, approximately $-1.58$.
This traveler order can be specified by the arrival time function $t(n) = -1.58 + n/300$.
Figure~\ref{fig:cost} shows the individual cost curves for several travelers, indicating that each traveler arrives at a time of minimum cost for their schedule delay parameters.
Our assertion that these $Q(t)$ and $t(n)$ functions represent equilibrium for all travelers (not just those marked in the figure) can be verified using the methods described in the following section. 

Using the equilibrium values for $Q(t)$ and $n(t)$, we can recover the other ways of representing the travel choice from Appendix~\ref{sec:choicerepresentation}:
\begin{itemize}
\item The travelers depart in order, so the ``permutation'' $\phi$ is simply the identity function $\phi(n) = n$.
\item To obtain the index of a traveler arriving at time $t$, we invert $t(n)$ to obtain $n(t) = 300(t + 1.58)$ for $t \in [-1.58, 0.42]$.
\item The departure time for a traveler arriving at time $t$ is computed as 
\begin{align*}
D(t) = t - Q(t) =
\begin{cases}
-0.225 (t + 1.58)^2 + 0.99t - 0.0158 & t \in [-1.58, 0] \\
-0.45 t^2 + 2.576 t - 0.579 & t \in [0, 0.42] 
\end{cases}
\,.
\end{align*}
\item Because the travelers depart in order, the cumulative arrivals by time $t$ are $\nu_A(t) = 300(t + 1.58)$ for $t \in [-1.58, 0.42]$; with $\nu_A(t) = 0$ beforehand and $\nu_A(t) = 600$ afterwards.
\item The departure time for the $n$-th traveler is
\[
\tau(n) = t(n) - Q(t(n)) =
\begin{cases}
-(2.5 \times 10^{-6}) n^2 + 0.0033n - 1.58 & n \in [0, 474] \\
-(5 \times 10^{-6}) n^2 + 0.0133 n - 5.78 & n \in [474, 600] \\
\end{cases}
\]
\item The cumulative departures by time $t$ are found by inverting $\tau(n)$; the solution to the resulting quadratic equation is straightforward but messy, so not shown here. Figure~\ref{fig:departure} illustrates the equilibrium cumulative arrivals and departures, and the corresponding departure rate.
\end{itemize}

\begin{figure*}[t!]
    \centering
    \begin{subfigure}[t]{\textwidth}
        \centering
        \includegraphics[width=10cm, height=5cm]{ 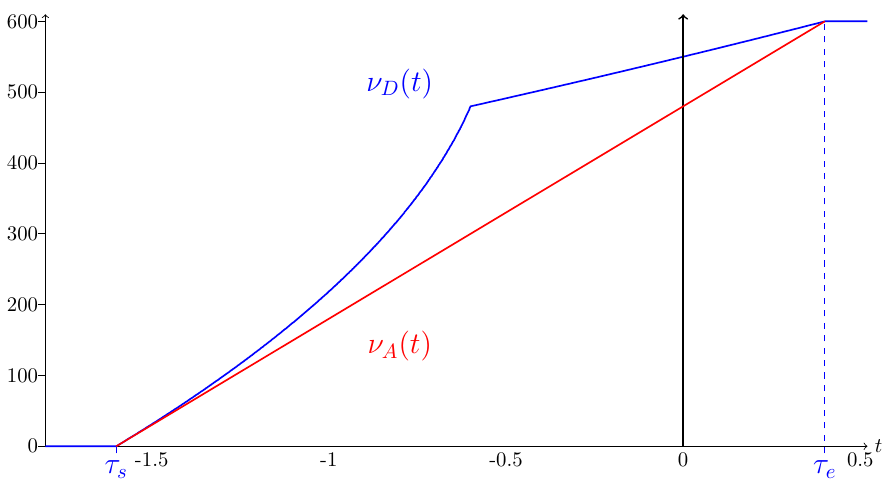}
        \caption{Cumulative arrivals and departures. }
    \end{subfigure}
    
    \begin{subfigure}[t]{0.95\textwidth}
        \centering
        \includegraphics[width=10cm, height=5cm]{ 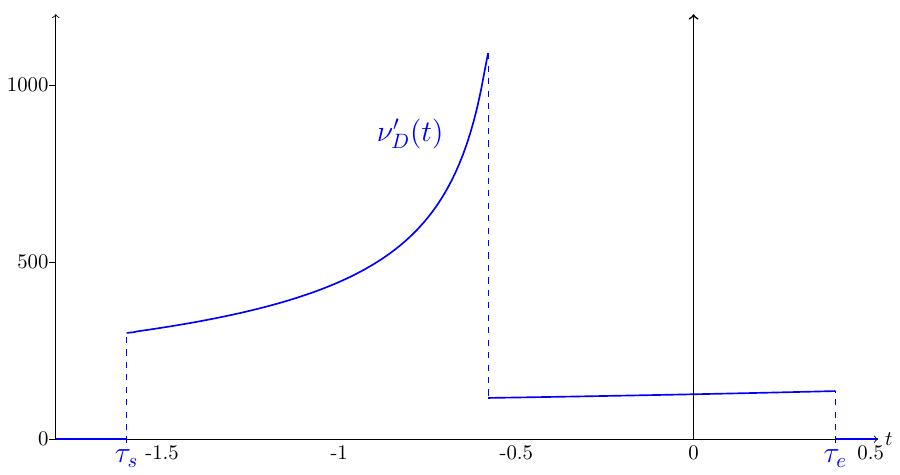}
        \caption{Departure rate.}
    \end{subfigure}
    \caption{Equilibrium departure and arrival dynamics.}\label{fig:departure}
\end{figure*}

The following section takes up the general question of equilibrium existence and uniqueness for our model.

The major assumption in our model is a ``monotonicity'' assumption, that a traveler index $n$ can be defined in such a way that $\beta(n)$ is increasing and $\gamma(n)$ is decreasing; that is, travelers with high early penalty weight must have low late penalty weight, and vice versa.
This assumption simplifies the analysis, compared to an arbitrary joint distribution of $\beta$ and $\gamma$, and we believe it is behaviorally plausible.
Vickrey's original formulation of the bottleneck model assumed that the schedule delay parameters $\beta$ and $\gamma$ arise from differences in monetary value of time under four different conditions: (1) at home, $\mu_h$; (2) while traveling, $\mu_t$; (3) at the destination prior to the preferred time, $\mu_{w1}$; and (4) at the destination after the preferred time, $\mu_{w2}$.
From this, we have $\alpha = \mu_t$ as the per-time unit cost of delay (we select units so that $\mu_t = 1$), $\beta = \mu_h - \mu_{w1}$ as the unit cost of early arrival, and $\gamma = \mu_{w2} - \mu_h$ as the unit cost of late arrival.
Traveler heterogeneity arises from many demographic factors, employment situations, and personal preferences, and can be modeled by distributions over $\mu_h$, $\mu_t$, $\mu_{w1}$, and $\mu_{w2}$.
One potential source of variability is the relative value of time spent at home $\mu_h$, compared to arriving early or late to work --- for instance, due to flexibility with remote work.
If variability in $\mu_h$ is the predominant source of heterogeneity, then we would expect $\beta = \mu_h - \mu_{w1}$ and $\gamma = \mu_{w2} - \mu_h$ to be negatively correlated.
In particular, assuming that heterogeneity in the other $\mu$ cost parameters is negligible compared to that in $\mu_h$ yields our monotonicity assumption.
To be sure, other assumptions on the sources of heterogeneity would yield different models, and we do not claim that our assumption is the only useful or valid way to model traveler heterogeneity in departure time choice.

\section{Equilibrium existence and uniqueness}
\label{sec:properties}

In this section we show existence of a unique equilibrium solution (functions $Q(t)$ and $t(n)$) to our model, assuming only that the functions $\beta(n)$ and $\gamma(n)$ are strictly monotone and continuous functions.
We do this by first considering a family of solutions where travelers depart in order of $n$ over a single interval, and with $Q(t)$ chosen to be a piecewise-differentiable function satisfying a first-order local optimality condition almost everywhere.
%Quote Index: R2_C6_a
Note that the equilibrium queueing profile may be non-differentiable, for example when the queue starts or ends, or at the peak of the queue)
We will show that one of these profiles satisfies the equilibrium principle, demonstrating existence.
In this first part of the section, we make no claim that all equilibria must share the ordering and other regularity properties of the equilibrium we constructed.
However, following the existence argument, we will show that these properties are indeed required, and that the solution we constructed initially is the only possible equilibrium solution, demonstrating uniqueness.

We start by proposing a solution of the following form, and will show that it is an equilibrium: travelers depart (and arrive) in order, over a single interval of duration $N/s$, with a queueing profile chosen to satisfy a first-order optimality condition.
Suppose that arrivals and departures occur with a single interval $[\tau_s, \tau_e]$, with $N_1 = -s \tau_s$ early arrivals and $N - N_1 = s\tau_e$ late arrivals.
In this case, \begin{align}
n(t) & = N_1 + ts, & t \in [\tau_s, \tau_e] \label{eqn:nt}, \\
t(n) &= (n - N_1) / s, & n \in [0, N] \label{eqn:tn}.
\end{align}
If $Q(t)$ is differentiable at $t$, then the derivative of the generalized cost as observed by traveler $n$ with respect to arrival time can be computed by:
\begin{equation}
\label{eqn:firstorder_C}
C_n'(t) =
\begin{cases}
Q'(t)-\beta(n) & t < 0 \\
Q'(t)+\gamma(n) & t > 0 \\
\end{cases},
\end{equation}
At equilibrium, this derivative should vanish at $t(n)$ for all $n \in [0,N]$, thus 
\begin{equation}
\label{eqn:firstorder_Q}
Q'(t) =
\begin{cases}
\beta(n(t)) & t < 0 \\
-\gamma(n(t)) & t > 0 \\  
\end{cases}.
\end{equation}

(Here we use the assumption that the support of $\beta$ is contained in (0,1), ensuring $Q'(t) < 1$ as is required by first-in, first-out ordering.)
Since $Q(\tau_s)=0$, we get a candidate for the equilibrium queue profile,
\begin{equation}
\label{eqn:queueintegral}
Q(t) =
\begin{cases}
\displaystyle
\int_{\tau_s}^t \beta(n(\alttime))~d\alttime = \frac{1}{s} \int_{0}^{n(t)} \beta(n)~dn & t \in [\tau_s, 0] \\
\displaystyle
Q(0) - \int_0^t \gamma(n(\alttime))~d\alttime = Q(0) - \frac{1}{s} \int_{N_1}^{n(t)} \gamma(n)~dn  & t \in (0, \tau_e]
\end{cases}
\,.
\end{equation}

We first consider the choice of $\tau_s$; for a queueing profile to be physically feasible, we must have $Q(\tau_s + N/s) = 0$, so that the queue has discharged fully over the assumed arrival interval.
\begin{lemma}
\label{lem:uniquestarttime}
There is a unique initial arrival time $\tau_s$ such that the queueing profile $Q(t)$ given by~\eqn{queueintegral} is physically feasible.
\end{lemma}
\begin{proof}
The starting time $\tau_s$ is linked to the number of early arrivals by $N_1 = -\tau_s s$, and it will be easier to work in terms of $N_1$.
Define the function
\begin{equation}
\label{eqn:Delta}
\Delta(N_1) = \int_{0}^{N_1} \beta(n)~dn - \int_{N_1}^{N} \gamma(n)~dn
\,,
\end{equation}
which represents the value of $Q$ at the time of last arrival $\tau_e$, according to~\eqref{eqn:queueintegral}.
For the queue to be physically feasible, we must have $Q(\tau_e) = \Delta(N_1) = 0$.
(If $\Delta(N_1) > 0$, more than $N$ vehicles exit the queue, whereas if $\Delta(N_1) < 0$ then $Q$ is negative near the end of the arrival window.)
This function is strictly increasing in $N_1$: if $N_1$ increases, we integrate over a larger portion of the $\beta$ distribution in the first term, and a smaller portion of the $\gamma$ distribution in the second term.
Further, $\Delta(0) < 0$ (because the first integral vanishes), and $\Delta(N) > 0$ (because the second does).
Since $\Delta$ is a continuous and strictly monotone function, there is a unique value of $N_1$ for which $\Delta(N_1) = 0$.
The choice $\tau_s = -N_1/s$ then satisfies the requirement that the queue is empty at the end of the interval $[\tau_s, \tau_e]$.
\end{proof}

With this lemma in hand, we now show that the proposed profile is an equilibrium solution.
The profile $Q(t)$ was constructed to be locally optimal for travelers, but we must show that each traveler's assigned arrival time is also globally optimal.
\begin{theorem}
\label{thm:existence}
The arrival time ordering $t(n)$ given by~\eqn{tn} and the queueing profile $Q(t)$ given by~\eqn{queueintegral}, with $\tau_s$ chosen according to Lemma~\ref{lem:uniquestarttime}, satisfy the equilibrium principle.
\end{theorem}
\begin{proof}
The queueing profile $Q(t)$ constructed in~\eqn{queueintegral} is continuously differentiable in the open intervals $(\tau_s,0)$ and $(0,\tau_e)$.
Therefore, the first-order equilibrium conditions, $C'_n(t(n))=0$, are satisfied for almost all $n$.

For any traveler $n$, the derivative of their cost profile at any arrival time $\alttime$ is
\begin{equation}
\label{eqn:costder}
C'_n(\alttime) =
\begin{cases}
Q'(\alttime) - \beta(n) & \alttime < 0 \\
Q'(\alttime) + \gamma(n) & \alttime > 0
\end{cases}
\,.
\end{equation}
This derivative vanishes if $\alttime = t(n)$.
If the traveler arrives early (so $t(n) < 0$) and $\alttime < t(n)$, then $C'_n(\alttime) < 0$ because $\beta$ is an increasing function: $C'_n(\alttime) = Q'(\alttime) - \beta(n) < Q'(\alttime) - \beta(n(\alttime)) = 0$, with the last equality following from the necessary condition~\eqn{firstorder_Q}.
If $t(n) < \alttime < 0$, then $C'_n(\alttime) > 0$ for the same reason: $C'_n(\alttime) = Q'(\alttime) - \beta(n) > Q'(\alttime) - \beta(n(\alttime)) = 0$.

If $\alttime > 0$, then $C'_n(\alttime)$ is still positive: $C'_n(\alttime) = Q'(\alttime) + \gamma(n) > Q'(\alttime) + \gamma(n(\alttime)) = 0$ because $\gamma$ is a decreasing function.
We therefore conclude that $C_n$ is strictly decreasing for $\alttime < t(n)$, and strictly increasing for $\alttime > t(n)$, so $t(n)$ is globally optimal, not just locally. The same argument holds for travelers who arrive late, \emph{mutatis mutandis}.

For a traveler arriving at $t = 0$, equations~\eqn{firstorder_Q} and~\eqn{costder} show that $C_{n(0)}(\alttime)$ is decreasing for $\alttime < 0$ and increasing for $\alttime > 0$, so in all cases travelers arrive at a globally optimal time.
\end{proof}

We now show that this departure order and queueing profile are the only ones representing equilibrium.
Our argument is distributed over several lemmas, and takes the following form:
\begin{itemize}
\item We define a schedule cost gap function, and show it is strictly increasing, a property helpful later in the argument. (Note: this is similar to the Monge property discussed by \cite{aka21}, but in our study both dimensions are continuous.) 
\item In any equilibrium solution, travelers must arrive in the order $t(n)$ given by~\eqn{tn} (from least to greatest $\beta$, and from greatest to least $\gamma$).
\item In any equilibrium solution, travelers must arrive in some single interval $[\tau_s, \tau_e]$ of length $N/s$, with the bottleneck discharging at capacity during the entire interval.
\item Within such an interval, the only possible equilibrium profile $Q(t)$ is one of the form~\eqn{queueintegral}, given the value of $\tau_s$.
\item There is exactly one value of $\tau_s$ which is feasible with the queueing dynamics.
As a result, the equilibrium presented above is unique.
\end{itemize}

Denote the conventional schedule cost function by
\begin{equation}
\label{eqn:scheduledelay}
\sigma_n(t) = \beta(n) [-t]^+ + \gamma(n) [t]^+ 
\,,
\end{equation}
and define the schedule cost gap function as
\begin{equation}
\label{eqn:scheduledelaygap}
\delta_{t_1, t_2}(n) =  \sigma_n(t_1) - \sigma_n(t_2)
\,.
\end{equation}

\begin{lemma}
\label{lem:Monge}
The schedule cost gap function $\delta_{t_1, t_2}(n)$ is strictly increasing in $n$, for any $t_1 < t_2, t_1,t_2\in[\tau_s,\tau_e]$.
\end{lemma}
\begin{proof}

Considering four possible cases:
\begin{description}
\item[Case I:] $t_1 < t_2 \leq 0$.
\begin{equation}
\label{eqn:scheduledelaygap1}
\delta_{t_1, t_2}(n) =  \sigma_n(t_1) - \sigma_n(t_2) = \beta(n) \cdot \left[ t_2 - t_1 \right]
\,.
\end{equation}
The difference $t_2 - t_1$ is positive and $\beta(n)$ is strictly increasing in $n$, hence  $\delta_{t_1, t_2}(n)$ is strictly increasing in $n$.
\item[Case II:] $0 \leq t_1 < t_2$.
\begin{equation}
\label{eqn:scheduledelaygap2}
\delta_{t_1, t_2}(n) =  \sigma_n(t_1) - \sigma_n(t_2) = \gamma(n) \cdot \left[ t_1 - t_2 \right]
\,.
\end{equation}
The difference $t_1 - t_2$ is negative and $\gamma(n)$ is strictly decreasing in $n$, hence  $\delta_{t_1, t_2}(n)$ is strictly increasing in $n$.
\item[Case III:] $t_1 < 0 \leq t_2$.
\begin{equation}
\label{eqn:scheduledelaygap3}
\delta_{t_1, t_2}(n) =  \sigma_n(t_1) - \sigma_n(t_2) = -\beta(n) t_1 - \gamma(n) t_2 
\,.
\end{equation}
 $\beta(n)$ is strictly increasing in $n$ and $- t_1$ is positive, hence  $- \beta(n) t_1$ is strictly increasing in $n$.
Similarly, $- t_2$ is non-positive and $\gamma(n)$ is strictly decreasing in $n$, hence  $- \gamma(n) t_2 $ is non-decreasing in $n$.
Therefore,  $\delta_{t_1, t_2}(n)$ which is the sum of the two terms is strictly increasing in $n$. 

\item[Case IV:] $t_1 \leq 0 < t_2$.
\begin{equation}
\label{eqn:scheduledelaygap3}
\delta_{t_1, t_2}(n) =  \sigma_n(t_1) - \sigma_n(t_2) = -\beta(n) t_1 - \gamma(n) t_2 
\,.
\end{equation}
 $\beta(n)$ is strictly increasing in $n$ and $- t_1$ is non-negative, hence  $- \beta(n) t_1$ is non-decreasing in $n$.
Similarly, $- t_2$ is negative and $\gamma(n)$ is strictly decreasing in $n$, hence  $- \gamma(n) t_2 $ is strictly increasing in $n$.
Therefore,  $\delta_{t_1, t_2}(n)$ which is the sum of the two terms is strictly increasing in $n$. 
\end{description}
\end{proof}

\begin{lemma}
\label{lem:monotone}
At equilibrium, $t(n)$ and $n(t)$ are strictly increasing functions.
\end{lemma}
\begin{proof}
The generalized cost can be expressed as $C_{n}(t) = Q(t) + \sigma_n(t)$.
By the equilibrium condition,  we have 
$C_{n(t)}(t) \leq C_{n(t)}(\alttime)$ for any alternative arrival time $\alttime$, 
which can be written as 
$Q(t) + \sigma_{n(t)}(t) \leq  Q(\alttime) + \sigma_{n(t)}(\alttime) $ 
or 
$\delta_{t, \alttime}(n(t)) \leq  Q(\alttime) - Q(t)$ and 
$\delta_{\alttime, t}(n(t)) = - \delta_{t, \alttime}(n(t)) \geq  Q(t) - Q(\alttime)$.
Therefore, for any $t_1<t_2$, 
$$\delta_{t_1, t_2}(n(t_1)) \leq  Q(t_2) - Q(t_1) \leq \delta_{t_1, t_2}(n(t_2))\,,$$
implying that $\delta_{t_1, t_2}(n(t_1)) \leq \delta_{t_1, t_2}(n(t_2))$. Since $\delta_{t_1, t_2}(n)$ is strictly increasing in $n$, and $n(t_1) \neq n(t_2)$, we must have $n(t_1)<n(t_2)$.
Since $t(n)$ is the inverse of $n(t)$, we conclude that $t(n)$ is strictly increasing as well.
\end{proof}

Our next result shows that the support of $Q$ is an interval of the form $(\tau_s, \tau_e)$, with $\tau_s < 0 < \tau_e$ and $\tau_e = \tau_s + N/s$.

\begin{lemma}
At equilibrium, the support of $Q$ is a single interval of width $N/s$, containing the preferred arrival time 0.
\end{lemma}
\begin{proof}
We first show that the support of $Q$ is a single interval.
Assume not.
Then there are times $t^*$, $t_1$, and $t_2$ such that $t_1 < t^* < t_2$ and $Q(t_1) > 0$, $Q(t^*) = 0$, $Q(t_2) > 0$.
If $t^* \leq 0$, consider the traveler $n(t_1)$ arriving at $t_1$, experiencing a cost $Q(t_1) + \beta[-t_1]^+$.
As $Q(t^*) < Q(t_1)$ and $t_1 < t^*$, we have $C_{n(t_1)}(t^*) < C_{n(t_1)}(t_1)$, violating the equilibrium condition~\eqn{equilibrium}.
If $t^* \geq 0$, we repeat the same argument for the traveler $n(t_2)$ arriving at $t_2$.

The support of $Q$ is therefore a single interval $(\tau_s, \tau_e)$.
It must contain the preferred arrival time; 
if $\tau_s \geq 0$, then any traveler arriving at  $t \in (\tau_s, \tau_e)$ can strictly reduce their cost by arriving at time zero instead, and similarly if $\tau_e \leq 0$.   
We now show that this interval has width $N/s$, that is, that the queue discharges at capacity during the entire interval in which travelers are departing.
As there are $N$ travelers, the interval cannot have width greater than $N/s$, so the only other option is that it is smaller.
If this is the case, a total of $s(\tau_e - \tau_s) < N$ travelers arrive during the interval with a queue, so there must be an interval of travelers arriving either before $\tau_s$ or after $\tau_e$.
In the first case, consider a traveler arriving prior to $\tau_s$.
This traveler can strictly reduce their cost by departing at $\tau_s$, since $Q(\tau_s) = 0$.
The second case is analogous.
\end{proof}

We now show that an equilibrium queueing profile must have a specific form.
\begin{lemma}
At equilibrium, $Q(t)$ must be given by equation~\eqn{queueintegral} for some value of $\tau_s$.
\end{lemma}
\begin{proof}
Given $\tau_s$ and any queue evolution profile $Q(t)$ with support $(\tau_s, \tau_e)$, define the gap function
\begin{equation}
\label{eqn:queuegap}
G(t) =
\begin{cases}
\displaystyle
Q(t) -  \int_{\tau_s}^{t} \beta(n(\alttime))~d\alttime & t \in [\tau_s, 0] \\
\displaystyle
Q(t) - Q(0) + \int_{0}^{t} \gamma(n(\alttime))~d\alttime  & t \in (0, \tau_e] 
\end{cases}
\,.
\end{equation}

We show that in any equilibrium solution the gap function $G(t)$ is uniformly zero.
If not, then $G(t) \neq 0$ for some $t$.
In this case, there must exist an interval $[t_1,t_2]$ with non-zero average slope (since $G(\tau_s)=0$).
Without loss of generality, assume that the slope is positive, that is, there is some $\epsilon > 0$ such that $G(t_2)-G(t_1) > \epsilon (t_2-t_1)$.
The interval can be chosen to be completely early or completely late.
Assume it is early (so $t_1 < t_2 \leq 0$); the proof is analogous for a late interval.

The chosen interval can be divided into a finite number of sub-intervals, so that the range of $\beta(n(t))$ within each sub-interval is at most $\epsilon/2$.
The average slope of $G(t)$ must be greater than $\epsilon$ in at least one of these sub-intervals.
Let $[t_3,t_4]$ be one such interval of length $\delta = t_4-t_3$, satisfying $G(t_4)-G(t_3) > \epsilon \delta$ and $\beta(n(t_4))-\beta(n(t_3)) < \epsilon/2$.
As a result,
\begin{align}
\label{eqn:queuegapbound}
\displaystyle
0 < \beta(n(t_4))\delta - \int_{t_3}^{t_4} \beta(n(\alttime))~d\alttime < \frac{\epsilon \delta }{2}
\,,
\end{align}
and therefore
\begin{align}
C_{n(t_4)}(t_4) - C_{n(t_4)}(t_3) &= Q(t_4) - Q(t_3) - \beta(n(t_4))  \delta  \\
&> Q(t_4) - Q(t_3) - \int_{t_3}^{t_4} \beta(n(\alttime))~d\alttime - \frac{\epsilon \delta}{ 2} \\
&= G(t_4) - G(t_3) - \frac{\epsilon \delta }{ 2}
> \frac{\epsilon \delta }{ 2} >0 
\,,
\end{align}
which contradicts the equilibrium condition.
\end{proof}

\begin{theorem}
\label{thm:uniqueness}
There is exactly one equilibrium queueing profile $Q(t)$ and traveler order $t(n)$.
\end{theorem}
\begin{proof}
The previous lemmas show that $t(n)$ and $Q(t)$ must take the forms given by equations~\eqn{tn} and~\eqn{queueintegral}, for some value of $\tau_s$.
By Lemma~\ref{lem:uniquestarttime}, there is a single value of $\tau_s$ which is physically feasible.
Therefore, the equilibrium used to demonstrate existence in Theorem~\ref{thm:existence} is the only one.
\end{proof}

\section{Equilibrium instability: discrete dynamics}
\label{Sec:InstabilityDiscrete}

The previous section established the existence of a unique equilibrium solution to our model, given any problem instance ($F_\beta$, $F_\gamma$, $N$, $s$).
We now show that this equilibrium solution is unstable for a wide family of ``reasonable'' day-to-day dynamic processes, so that the delay process does not converge to equilibrium even when starting arbitrarily close to an equilibrium solution.
Informally, by ``reasonable'' we mean that the evolution in queue profile reflects myopic and selfish behavior, responding to local ``pressure'' based on the derivative of the cost function at their time of arrival.
The subsections below formalize this approach.
Subsection~\ref{sec:individualresponse} provides motivation by considering how an individual traveler might respond in a non-equilibrium setting.
Subsection~\ref{sec:dynamicsclasses} proposes several classes of dynamics based on this principle, and in particular suggests that order-preserving, local, pressure-related (OLP) dynamics (one of these classes) are particularly worthy of investigation.
Subsection~\ref{sec:olpinstability} provides formal definitions of OLP dynamics and stability, and takes up the question of whether any such dynamic process is stable.

\subsection{Traveler response to non-equilibrium solutions}
\label{sec:individualresponse}

We begin by considering how a traveler might respond after experiencing a non-equilibrium travel day in our model, when making a choice of departure time on the subsequent day.
In this subsection, we consider only the response of an individual traveler; the following subsections will discuss dynamic processes in which many travelers simultaneously respond.
We assume that the traveler is selfish, aiming to reduce their generalized cost without concern for that of any other traveler.
We also assume that the traveler is myopic and does not anticipate the impact of their response on the costs they experience on the next day.
If this traveler is the only one changing their choice (as considered in this subsection), they have an infinitesimal impact on the queueing and delay profiles and can treat them as constant; the following subsection discusses this assumption more when all travelers are responding.

\begin{figure*}[t!]
    \centering
        \includegraphics[scale=0.95]{ 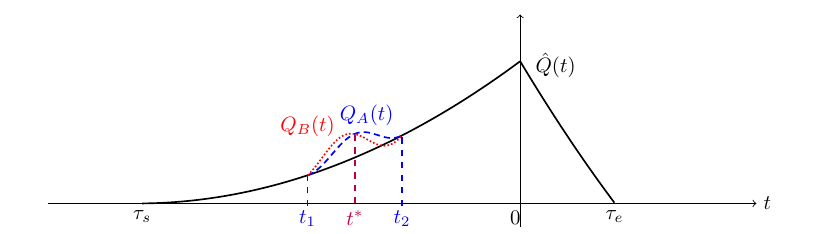}
    \caption{The solid black line is the equilibrium profile $\hat{Q}(t)$ from Figure~\ref{fig:dep_queue}.  The blue and red lines are perturbed queueing profiles.    }
    \label{fig:stability}
\end{figure*}
\begin{figure}
  \centering
    \includegraphics[scale=0.6]{ 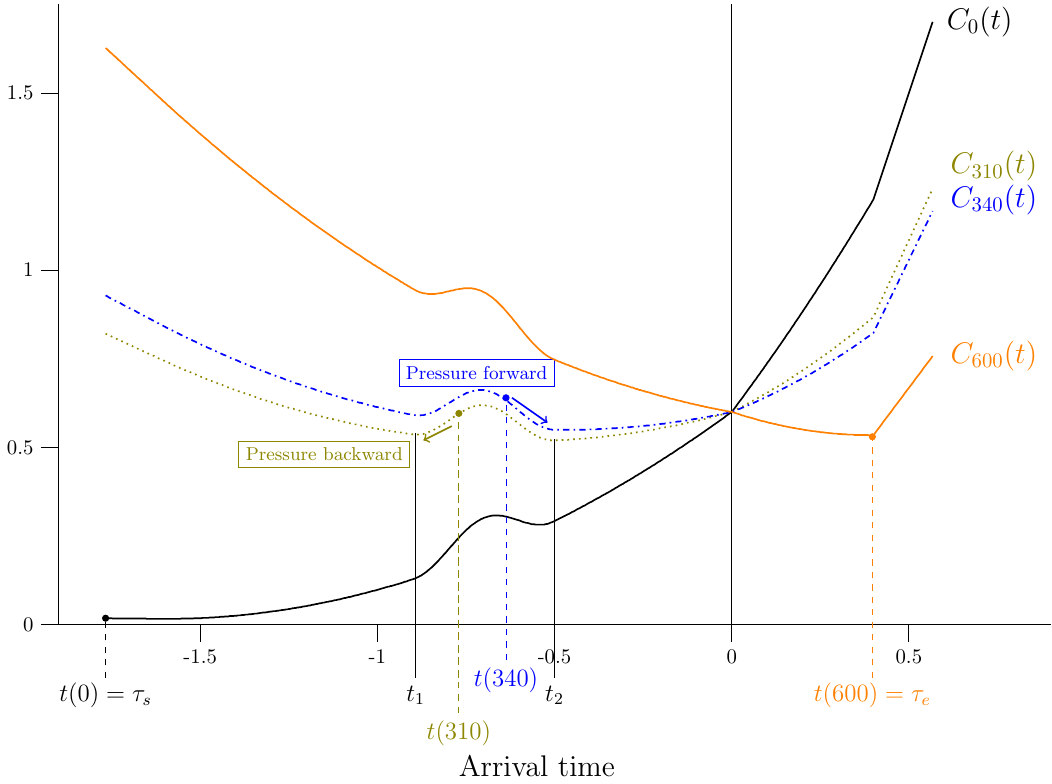}
    \caption{The cost functions $C_n(t)$ after a small perturbation away from equilibrium. Travelers within the perturbation interval are pressured to move in the direction of a new local minimum.}\label{fig:cost_per_rev}
\end{figure}

As an example, consider the scenario in Figure~\ref{fig:stability}.
In this figure, the solid black line is the equilibrium queueing profile $\hat{Q}(t)$ assuming linear heterogeneity in delay parameters (the example outlined at the end of Section~\ref{sec:model}), shown both in terms of a traveler's arrival time and a traveler's departure time.
Under the equilibrium profile, each traveler arrives at a time of minimum cost and has no pressure to depart earlier or later.
The dashed blue line $Q_A(t)$ and dotted red line $Q_B(t)$ are perturbed queueing profiles which differ 
from the equilibrium solution in a subinterval of the early arrivals.
We perturb only the queueing profile in this example, and assume that travelers remain ordered by increasing $\beta$ (and decreasing $\gamma$), as at equilibrium, for reasons discussed in the next subsection.

Within this interval, travelers are no longer experiencing minimum costs, and can reduce their cost by arriving either earlier or later (we call this ``pressure''); outside of the perturbation interval, travelers continue to experience minimum costs and face no pressure to move.
Figure~\ref{fig:cost_per_rev} illustrates this point for the perturbation $Q_A$.
This figure plots the cost profiles of different travelers, with a solid dot indicating the time at which they arrive.
The cost profile for traveler $n$, under any queueing profile $Q$, is expressed as
\begin{equation}
\label{eqn:perturbcost_rev}
C_n(t \mid Q) = Q(t) + \beta(n) [-t]^+ + \gamma(n) [t]^+
\,.
\end{equation}

The black line $C_0$ shows the cost 
traveler 0 (with minimum $\beta$ and maximum $\gamma$) would have experienced if they had arrived at any time between $\tau_s$ and $\tau_e$; the minimum is at $\tau_s$, their arrival time, so there is no pressure to arrive earlier or later.
Even though their arrival cost has been perturbed between $t_1$ and $t_2$, the minimum point is unchanged.
The gold and blue curves show cost profiles for travelers 310 and 340.
These travelers arrive during the perturbation window, and their arrival times are no longer times of minimum cost.
They therefore experience pressure to arrive at a different time.
The labels in the figure indicate a ``local'' pressure direction.
The distinction between local and global pressure is discussed more in the following subsection; what is important is that the traveler prefers to arrive at a different time than the one in which they actually do, and will adjust their departure time in a way that, given the current queueing profile $Q$, would lead to arrival time with lower total cost.

\subsection{Classes of dynamic processes}
\label{sec:dynamicsclasses}

The previous subsection used a small example to motivate the idea of ``pressure,'' that in a non-equilibrium scenario some travelers see opportunities to reduce their cost by adjusting their departure time.
We now discuss several classes of dynamic processes that suggest how to represent the evolution of the system when many drivers are simultaneously making such adjustments.
Our descriptions in this subsection remain mostly intuitive, and serve to motivate the formal definition and analysis of a specific class of dynamic process (order-preserving local pressure-related dynamics) in the following subsection.
Given a traveler departure order and a queueing profile, a dynamic process specifies each traveler's departure time on the subsequent day.
As discussed in Appendix~\ref{sec:choicerepresentation}, knowing each traveler's departure time is equivalent to knowing the departure order and queueing profile, and vice versa; so these dynamic processes can equivalently be understood as mapping one day's departure times to the next day's, or from one day's departure order and queueing profile to the next day's departure order and queueing profile, and so forth.
For purposes of exposition, it is helpful to motivate these dynamic processes in terms of departure time choice (where choice behavior is most clearly expressed), but for purposes of analysis the equivalent representation in terms of departure order and queueing profile will be more convenient.

We retain the assumptions that drivers are selfish and myopic, making their choices based on the current queueing profile $Q$ and not any prediction of how $Q$ might evolve based on others' choices.
That is, we assume that drivers choose a departure time, assuming that their arrival time will be based on the queueing profile $Q$ they just experienced.
We believe this assumption is justified, because alternatives would require travelers to have a significant amount of hard-to-observe information (the bottleneck capacity, the number of other players, the distribution of their schedule delay parameters, and knowledge of their strategic behavior).
By contrast, the myopic assumption requires only that drivers are aware of their own personal arrival time costs --- they need only know the queueing profile, and nothing more.
The myopic assumption also allows us to easily pass back and forth between departure times and arrival times (drivers act as if they will continue to differ exactly by $Q$).
(It is helpful at times to refer to both perspectives; the fundamental travel choice is departure time, but schedule delay is based on arrival time.)
Finally, the myopic assumption is consistent with dynamics commonly assumed in day-to-day traffic assignment models \citep[see][for analysis of several common update dynamics]{gou18}, for instance, in the Braess paradox, drivers switch routes intending to reduce their individual cost, even as they fail to anticipate that others switching routes as well will ultimately increase their cost.

We say a dynamic process is \emph{pressure-related} if \emph{each traveler chooses a departure time in such a way that, assuming $Q$ were to remain the same, their cost on the next day would be no greater than that on the current day.}   That is, if $\hat{\tau}$ and $\tau$ are the departure times of traveler $n$ on the current and next days, we require $C_n(t \mid Q) \leq C_n(\hat{t} \mid Q)$, where $t = \tau + Q(t)$ and $\hat{t} = \hat{\tau} + Q(\hat{t}$).
This class of dynamic processes is very broad, including best-response dynamics, and any other which moves each traveler to a departure time no worse than their current one.
Most of the dynamics discussed in Section~\ref{sec:litreview} are pressure-related.

We say that a pressure-related dynamic process is \emph{local pressure-related} if the direction of change in departure time corresponds to the derivative in the cost profile at the time of arrival in $Q$.
Refer again to Figure~\ref{fig:cost_per_rev}, and consider travelers 310 (gold profile) and 340 (blue profile).
When traveler 310 arrives at the destination, their cost profile is increasing ($C'_{310} > 0$).
They therefore experience local pressure \emph{backward}: if they were to arrive at a slightly earlier time, their cost would decrease; and to do so, they must depart at an earlier time.
On the other hand, traveler 340 experiences local pressure \emph{forward}: $C_{340}$ is decreasing at their time of arrival ($C'_{340} < 0$), so if they were to arrive at a slightly later time (requiring slightly later departure), their cost would decrease.

We again emphasize that these assessments are made myopically, assuming $Q$ remains unchanged and without considering the strategic behavior of the other travelers in the system, and it is quite possible that decisions made with the intent of arriving earlier or later may not actually have the intended effect.

There are two reasons why local pressure-related dynamics are interesting to study.
Behaviorally, they require drivers to have a minimum of information: all that is needed is to know how the queue is evolving at the time of arrival (``is congestion getting better or worse when I arrive?''), and nothing about queue behavior at other time points.
We also expect local pressure-related dynamics to have more favorable stability properties than general pressure-related dynamics, because travelers are only accounting for cost changes in the neighborhood of their time of arrival, rather than over the entire departure window.
With local dynamics, only travelers arriving during the perturbation window face pressure to change their departure time.
By contrast, in general pressure-related dynamics, a perturbation during a small time window may induce many travelers throughout the entire departure window to adjust their choice.

We finally say that a local pressure-related dynamic process is \emph{ordered local pressure-related} (OLP) if the departure order $n(t)$ remains unchanged from one day to the next.
We expect OLP dynamic processes to have even more favorable stability properties, because we know that the equilibrium solution in our model exhibits an ordering property (Lemma~\ref{lem:monotone}).
If all dynamic processes which preserve this necessary ordering property can be shown to be unstable, then we would expect a broader class of dynamic processes which cannot enforce it to be unstable as well.
Behaviorally, the ordering assumption ensures that the departure time changes for travelers with similar departure times today are not so different that their departure order would switch tomorrow.

In local pressure-related dynamics, a traveler $n$ arriving at a time when $C'_n > 0$ faces ``backward pressure'' to depart earlier, and one arriving when $C'_n < 0$ faces ``forward pressure'' to depart later.
In OLP dynamics, we additionally have the property that all travelers will continue to \emph{arrive} at the same time, because the departure order is unchanged, and because the first and last travelers face no queue (and therefore no pressure to change their departure time.).
Therefore, backward pressure corresponds to an increase in the queueing delay experienced (the traveler departed earlier, but arrived at the same time), whereas forward pressure corresponds to a decrease (the traveler departed later, but arrived at the same time).

Consider again travelers 310 and 340 in Figure~\ref{fig:cost_per_rev}. 
Traveler 310 experiences ``backward pressure,'' as their cost profile is increasing at their moment of arrival, and traveler 340 experiences ``forward pressure.''
On the next day, traveler 310 will depart earlier, and traveler 340 will depart later, intending to reduce their cost.
However, other travelers are simultaneously changing their departure choices, and in an OLP dynamic these changes are such that the traveler order (and thus the arrival time) are the same.
The effect is for traveler 310 to experience a longer queue on the next day, and traveler 340 to experience a shorter queue.
Referring to Figure~\ref{fig:stability}, the perturbed queueing profile $Q_B$ is obtained by applying OLP dynamics to the perturbation $Q_A$; travelers departing before the ``peak'' of the perturbation (such as traveler 310) experience a higher queueing delay, and those departing after the peak (such as traveler 340) experience a lower delay.
Conversely, $Q_A$ cannot be obtained from $Q_B$ by applying OLP dynamics; to do so would require travelers to move contrary to the direction given by the local pressure.

To summarize, OLP dynamics have the property that, at each arrival time $t$, the change in the queue length experienced by a traveler arriving at $t$ has the same sign as the derivative of their cost profile at $t$.

\subsection{Instability of ordered local pressure-related dynamics}
\label{sec:olpinstability}

This subsection formally defines ordered local pressure-related dynamics and stability, and analyzes the stability properties of the former.

An ordered dynamic process is a mapping $Y: \mc{Q} \rightarrow \mc{Q}$, where $\mc{Q}$ is the space of admissible queueing profiles $Q(t)$, to be specified below.
The space of all ordered dynamic processes $\mc{Y}$ is very large and contains ``trivial'' dynamics, such as one that maps every queueing profile directly to equilibrium.
To make interesting statements about stability, we will restrict our attention to subsets of $\mc{Y}$.
\begin{definition}
An ordered dynamic process $Y$ is \emph{ordered local pressure-related (OLP)} if, for all $Q_A \in \mc{Q}$, for $Q_B = Y(Q_A)$, and for all $t$ where $Q_A$ is differentiable, we have
\begin{equation}
\label{eqn:pressure}
\left[ Q_B(t) - Q_A(t) \right] \cdot C'_{n(t)} (t \mid Q_A) \geq 0 
\,,
\end{equation}
where $n(t)$ is the index of the traveler arriving at time $t$ (given by the assumed order).
\end{definition}

We now introduce regularity conditions on the class of admissible delay profiles $\mc{Q}$, which will include the equilibrium and any admissible perturbations.
Let $\tau_s$ and $\tau_e$ be the unique start and end times  corresponding to the equilibrium for the given problem instance, and $\mc{T} = [\tau_s, \tau_e]$ the equilibrium arrival interval.
\begin{definition}
\label{def:mcQ}
Let $\mc{Q}$ denote a class of queueing delay profiles such that any $Q \in \mathcal{Q}$ satisfies:
\begin{enumerate}
\item The support of $Q$ is $\interior \mc{T}$, and $Q$ is strictly positive on its support.
\item $Q$ is continuous and piecewise continuously differentiable on $\mc{T}$. 
\item For any $t_1 < t_2$, we have $Q(t_2) - Q(t_1) < t_2 - t_1$.
\end{enumerate}
\end{definition}
The first condition requires that the departure window coincide with that at equilibrium.
OLP dynamics preserve this property; the first and last travelers experience no queueing delay and thus face no pressure to move.
The third condition ensures the queue profile is physically plausible and represents first-in, first-out behavior (ensuring $Q'(t) < 1$ where $Q$ is differentiable).
Appendix~\ref{sec:olpexistence} gives an example of a pressure-related dynamics for this class of queueing profiles.
Our analysis does not require any specific form of the dynamics beyond the pressure condition~\eqn{pressure}; for instance, we can allow travelers at different departure times to respond in stronger or weaker ways to the pressure.

A basic condition for such a process to be stable is that the equilibrium $\hat{Q}$ be a fixed point of $Y$, that is, $Y(\hat{Q}) = \hat{Q}$.
To assess stability, we use the $L_\infty$ norm to measure the distance between any two delay profiles $Q_1, Q_2 \in \mc{Q}$:
\begin{equation}
\label{eqn:distance}
d(Q_1, Q_2) = \max_{t \in \mc{T}} |Q_1(t) - Q_2(t)| \\
\,.
\end{equation}

With these definitions in mind, we present the stability notion we will use in the remainder of this section:
\begin{definition}
A mapping $Y : \mc{Q} \to \mc{Q}$ with a fixed point  $\hat{Q}$ (i.e., $Y(\hat{Q})=\hat{Q}$) is said to be \emph{converging} if there exists $\epsilon > 0$  such that, for any $Q$ satisfying $d(Q, \hat{Q}) < \epsilon$, the sequence $\{ Y^k(Q) \}$, obtained by iteratively applying $Y$ to $Q$ for $k$ times, converges to $\hat{Q}$.
\end{definition}\begin{definition}
An equilibrium $\hat{Q}$ is stable with respect to a class $\mc{K}$ of ordered dynamics ($\mc{K} \subseteq \mc{Y}$), all sharing a fixed point at $\hat{Q}$, if the class $\mc{K}$ contains a converging mapping $Y \in \mc{K}$. Conversely, the equilibrium is unstable with respect to the class $\mc{K}$ if the class does not contain a converging mapping.
\end{definition}
Note that our definition of stability is rather general: the existence of \emph{any} dynamics in the class $\mc{K}$ is sufficient, including a mapping chosen specifically for the given perturbation.

Our main result is that \emph{no} OLP dynamic will converge to the equilibrium, even from arbitrarily small perturbations.
To do so, we first discuss the existence of small perturbations.
Next, we evaluate the impact of OLP dynamics on the distance from equilibrium.
Finally, we conclude with the main instability theorem.

\begin{lemma}
\label{lemma:PerturbationExistence}
For any $\epsilon > 0$ there exists a queueing profile $Q \in \mc{Q}$ such that $0 < d(Q, \hat{Q}) \le \epsilon$.
\end{lemma}
\begin{proof}
    Choose a closed interval during the early arrivals period, $[t_1, t_2] \subset (\tau_s, 0)$, over which $Q$ is continuously differentiable, and choose $e(t)$ to be a continuously differentiable non-negative function which is nonzero on the interior of this interval and zero elsewhere. Let $\hat{d} = \max_{t \in [t_1, t_2]} \{\hat{Q}'(t) \}$ and $d = \max_{t \in [t_1, t_2]} \{e'(t)\}$; by property 3 of Definition~\ref{def:mcQ} we have $\hat{d} < 1$.  Let $m = \max (e(t))$, and consider the profile $Q = \hat{Q} + e \cdot \min \{ \epsilon/m, (1 - \hat{d})/d \}$. Clearly $0 < d(Q, \hat{Q}) \le \epsilon$ as required. Furthermore, $Q \geq \hat{Q} > 0$ over int~$\mc{T}$ and $Q' < 1$, thus satisfying the conditions of Definition \ref{def:mcQ}. 
\end{proof}

\begin{lemma}
\label{lemma:PerturbationDivergence}
Let $Y$ be an ordered local pressure-related dynamic with a fixed point at the equilibrium $\hat{Q}$, let $Q_A \in \mc{Q}$ be any non-equilibrium queueing profile (i.e., $d(Q, \hat{Q})>0$), and let $Q_B = Y(Q_A$).  Then $d(Q_B, \hat{Q}) \geq d(Q_A, \hat{Q})$.
\end{lemma}
\begin{proof}
Let $e = Q_A - \hat{Q}$ describe the initial perturbation from equilibrium. 
The distance from equilibrium is $ m = d(Q_A, \hat{Q}) = \max_{t \in \mc{T}} |e(t)| > 0$.
Denote the first time of maximal deviation by:
\begin{align}
    t^* = \min \left\{ t \in \mc{T} : |e(t)|= m \right\} 
\,.
\end{align}

Suppose first that $e(t^*)$ is positive, thus $e(t^*)=m>0$. We do not assume that $Q_A$ is differentiable at $t^*$, but since $Q_A$ is piecewise differentiable there exists $\delta>0$ such that $Q_A$ is differentiable over the open interval $(t^*-\delta, t^*)$. Since $e$ obtains its maximum at $t^*$ for the first time, $e'(t) > 0 \quad \forall t \in (t^*-\delta, t^*)$.
Note that
\begin{align}
\label{eqn:pressureexample}
C'_{n(t)}(t) = \hat{Q}'(t) + e'(t) - \beta(n(t)) = e'(t)
\  \forall t \in \hat{\mc{T}}(Q_A), t<0 \\
C'_{n(t)}(t) = \hat{Q}'(t) + e'(t) + \gamma(n(t)) = e'(t)
\  \forall t \in \hat{\mc{T}}(Q_A), t>0
\,,
\end{align}
with the last equality in both rows following because $\hat{Q}$ is an equilibrium satisfying~\eqn{firstorder_Q}.

Using \eqn{pressure} we get that
\begin{align}
\label{eqn:der_C}
C'_{n(t)}(t) &> 0 & \forall t \in (t^*-\delta, t^*), \\
Q_B(t) - Q_A(t) &\geq 0 & \forall t \in (t^*-\delta, t^*).
\end{align}
Since $Q_A$ and $Q_B$ are continuous, we conclude that 
\begin{align}
    Q_B(t^*) &\geq Q_A(t^*) \\
    d(Q_B, \hat{Q}) &\geq Q_B(t^*) - \hat{Q}(t^*) \geq Q_A(t^*) - \hat{Q}(t^*) = m = d(Q_A, \hat{Q}).
\label{eqn:d_increase}
\end{align}

If $e(t^*)<0$ a similar argument applies, with the appropriate modifications in \eqn{der_C}--\eqn{d_increase}.
\end{proof}

\begin{theorem}
\label{thm:unstable}
For any problem instance ($F_\beta$, $F_\gamma$, $N$, $s$), the unique equilibrium profile $\hat{Q}$ is unstable with respect to the class of OLP dynamics.
\end{theorem}
\begin{proof}
Suppose by contradiction that there exists $\epsilon > 0 $ and a pressure-related mapping $Y$ such that for any queueing profile $Q$ with $d(Q, \hat{Q}) \leq \epsilon$ an iterative application of the mapping converges to the equilibrium, i.e., $\lim_{k \to \infty} d(Y^k(Q), \hat{Q}) =0$.  According to Lemma \ref{lemma:PerturbationExistence} there exists a queueing profile $Q_0 \in \mc{Q}$ such that $d(Q, \hat{Q}) = \epsilon_1 < \epsilon$ where $\epsilon_1>0$. Let $Q_{k+1} = Y(Q_k)$ for $k \in \mathbb{N}$. By Lemma \ref{lemma:PerturbationDivergence} 
$d(Q_{k+1}, \hat{Q}) \geq d(Q_k, \hat{Q})$. Therefore, by induction, $d(Q_k, \hat{Q}) \geq \epsilon_1$ for all $k \in \mathbb{N}$, in contradiction to the convergence assumption.
\end{proof}

\section{Discussion} 
\label{Sec:Discussion}

The results of Section \ref{Sec:InstabilityDiscrete} can be connected to several related questions regarding stability in other contexts, including:
the original model of Vickrey with
homogeneous travelers; % 1
non-linear scheduling penalty functions;  %2 
the dynamics if `days' are considered as a continuous dimension; % 3
the connection between equilibrium uniqueness and stability; %4
 non-local or non-ordered dynamics; %5 
equilibrium models with continuous choice, not necessarily related to departure time; % 6
general joint distribution of scheduling penalty parameters; % 7
and more. A thorough evaluation of each extension remains a subject for future research. This section will discuss briefly our current understanding of these directions.

The case of homogeneous travelers can be viewed as a limiting case of strictly monotone functions $\beta(n)$ and $\gamma(n)$, when the slope approaches zero. In that sense, the results presented here apply to models that approximate the homogeneous case as closely as we want. When considering dynamics for the homogeneous case, since all travelers are indistinguishable, we can re-sort them according to \eqref{eqn:nt} after every iteration. If the resulting mapping satisfies \eqref{eqn:pressure}, then Lemma \ref{lemma:PerturbationDivergence} holds and the instability result applies. However, many plausible day-to-day dynamics for the homogeneous case do not satisfy \eqref{eqn:pressure}, and thus their stability should be analyzed by other approaches.

The case of non-linear scheduling penalty functions implies that $\beta$ and $\gamma$ depend not only on the traveler but also on the actual time. This adds a slight complication to the notation, but otherwise should not affect the analysis.

The pressure-related mappings studied in Section \ref{Sec:InstabilityDiscrete} can be seen as discrete day-to-day dynamics. Specifically, if $Q$ is the queue delay profile on day $k=0$, then $Y^k(Q)$ is the profile on day $k\geq 1$. Previous work has established instability of the equilibrium in classical departure time choice models using continuous day-to-day dynamics, e.g., \cite{iryo08} and \cite{iryo19}. In particular, the change in the arrival profile is represented as a differential equation, and does not jump as the discrete update formula of $Y(Q)$ satisfying the pressure dynamics condition \eqref{eqn:pressure}.
A continuous counterpart of our dynamics can be constructed as follows. Let $Y^0=Q$ denote an initial profile and denote by $Y^\tau$ the profile on day $\tau\geq 0$. The continuous daily change of the profile at any $t\in\mathcal{T}$ is given by the differential equation
\begin{align}
\frac{\partial }{\partial \tau}Y^\tau(t) =a(t) C'_{n(t)} (t \mid Y^\tau).
\end{align}
As in the proof of Theorem~\ref{thm:unstable}, a profile that is a small perturbation from the equilibrium will have $\frac{\partial }{\partial \tau}Y^\tau(t^*)=0$ at the maximal deviation at time $t^*$, and thus will not move back towards the equilibrium profile.

The key difference of this instability result from previous work is in the local nature of the pressure-dynamics; at every time $t$ the change of the profile is only determined by the derivative of the cost at time $t$. This takes inspiration from gradient-descent algorithms. In the context of queueing games such dynamics are known to converge to the Nash equilibrium in many cases \citep{Ravner25}. In contrast, the dynamics of \cite{iryo08} involve comparing the cost at $t$ with the cost any time $\alttime\in\mathcal{T}$. Using our notation this entails computing $C_{n(t)}(t)-C_{n(t)}(\alttime)$ for all $\alttime\in\mathcal{T}$, and then the change of the profile at $t$ is given by a (weighted) integral of all travelers that can reduce their cost by moving to $t$. One might expect local dynamics to have more favorable stability properties, because travelers only consider local changes to their arrival time in response to a small perturbation and do not consider all possible arrival times.  In this sense, the instability result of Theorem~2 is strong.

The basic logic of Lemma \ref{lemma:PerturbationDivergence}, separating the additional pressure related to a perturbation from the zero pressure associated with the equilibrium profile may be valid in other situations, even if the equilibrium is not unique, or if the choice is continuous but not related to departure time.

Our main results were obtained assuming a special form of ``monotone'' heterogeneity in schedule delay parameters.  The case of a general joint distribution of penalty parameters can be explored in several ways. One approach is to start from an equilibrium solution for the monotone model, and identify a class of two-dimensional joint distributions for which the same queuing profile satisfies the equilibrium condition. This ``reverse engineering'' approach is demonstrated explicitly below, for an example with linear penalty functions.  (A full analysis of this approach is beyond the scope of the current article.) Using this approach, we can construct examples where the $\beta$ and $\gamma$ do not have a monotone relationship, but yet the aggregate equilibrium queueing profile is the same as if they did. We may also anticipate conditions for pressure-related dynamics to be mainly local. However, demonstrating  existence or uniqueness and evaluating stability for the general case of two-dimensional distributions requires substantial additional analysis.

The linear problem instance described at the end of Section~\ref{sec:model} enables a relatively simple characterization of two-dimensional parameter distributions for which the same queuing profile corresponds to an equilibrium state. To explore this direction, note that since $n(t) = 300(t - \tau_s)$, the penalty parameters can be expressed as a function of time,
$\beta(n(t)) = 0.01 + 0.45 (t - \tau_s) $, and $\gamma(n(t)) = 3 - 0.9 (t - \tau_s) $.
Letting $C^* = C_{n(0)}(0 \mid Q) = 0.225 \tau_s^2 - 0.01\tau_s$ we get,
$$C_{n(t)}(t \mid Q) = \begin{cases}
Q(t) - \beta(n(t)) t  = C^* - 0.225 t^2 & t \in [\tau_s, 0] \\
Q(t) + \gamma(n(t)) t = C^* - 0.45 t^2 & t \in [0, \tau_e]
\end{cases}
\,.$$ 
For any early arrival at $t<0$, a late arrival with an equal cost occurs at $t'=-{0.5}^{0.5}t$, which is obtained by solving $C_{n(t)}(t \mid Q)=C_{n(t')}(t' \mid Q)$. Note that $t\approx-0.594>\tau_s$ for $t'=\tau_e \approx 0.42$,
where $n(-0.594)=295.8$ and $\beta(295.8) = 0.4537$.

Now consider a new traveler, of negligible weight (cannot influence the queue), with penalty parameters $\beta_0, \gamma_0$. Geometrically, each such traveler can be represented as a point in the plane depicted in Figure~\ref{fig:division}. The support of the distribution of our monotonic example is shown in the figure as a blue line. 
We can identify a division line, shown in green in Figure~\ref{fig:division}, along which each traveler is indifferent between two options. In our case, this line consists of three sections, as described in Table~\ref{tab:transline}. In this table, the three sections are defined by their range of $\beta_0$, listed in the first column. The conditions in the second column define the connection between $\beta_0$ and $\gamma_0$. The two alternative travel times, with equal cost, are listed in the rightmost columns. At the end point of the division line, where $\beta_0 = \beta(n(0))$ and $\gamma_0 = \gamma(n(0))$, the traveler has only one preferred arrival time, $t=0$.  
Note that sections A and C are linear, while section B is quadratic. Sections B and C are clearly visible in Figure~\ref{fig:division}, while section A is too small to illustrate in this figure. 

\begin{table}
\begin{tabular}{|c|cccc|}
\hline
Section & $\beta_0$ Range & Conditions & Early AT & Late AT \\
\hline
\hline
A & $[0, 0.01]$ & $- \beta_0 \tau_s = \gamma_0 \tau_e$ & $\tau_s$ & $\tau_e$ \\
\hline
B & $[0.01,0.4537]$ &  
\begin{tabular}{c}
$\beta_0 = \beta(n(t))$\\
$\gamma_0 =  \sfrac{\left( C^* - 0.225 {t}^2 \right)}{\tau_e}$
\end{tabular}
& $t$ & $\tau_e$ \\
\hline
C & $[0.4537, \beta(n(0))]$ &
\begin{tabular}{l}
$\beta_0 = \beta(n(t))$\\
$\gamma_0 = \gamma(n(t'))$\\
$t'=-{0.5}^{0.5} t$
\end{tabular} &
$t$ & $t'$ \\
\hline
\end{tabular}
AT - Arrival Time
\caption{Division line equation for the linear case example.}\label{tab:transline}
\end{table}

The information in Table~\ref{tab:transline} defines the single preferred arrival time for a traveler with any other combination of penalty parameters, according to three cases:

\begin{enumerate}
\item 
If $\beta_0, \gamma_0$ is on the division line, associated with early arrival time at $t$, and $\gamma_1 > \gamma_0$, then a traveler with parameters $\beta_0, \gamma_1$ will prefer to arrive at $t$. These points fall along a vertical ray upwards.  Figure~\ref{fig:division} shows illustrative vertical rays that are intersecting with the support line (Case 1), while rays from section A of the division line are not shown.
\item  
If $\beta_0, \gamma_0$ is on the division line, associated with late arrival time at $t'$, and $\beta_1 > \beta_0$, then a traveler with parameters $\beta_1, \gamma_0$ will prefer to arrive at $t'$. These points fall along a horizontal ray to the right. Figure~\ref{fig:division} shows illustrative horizontal rays that are intersecting with the support line (Case 2), while rays from sections A and B of the division line are not shown.
\item 
Any traveler with parameters $\beta_1 > \beta(n(0))$ and $\gamma_1 > \gamma(n(0))$ will prefer to arrive on time at $t=0$. These points fall in a shifted quadrant. In Figure~\ref{fig:division} this quadrant is the empty region bounded by dashed blue lines (Case 3), with a corner at the top right end of the division line.
\end{enumerate}

%%%
\begin{figure}[t!]
  \centering
    \includegraphics[scale=0.7]{ 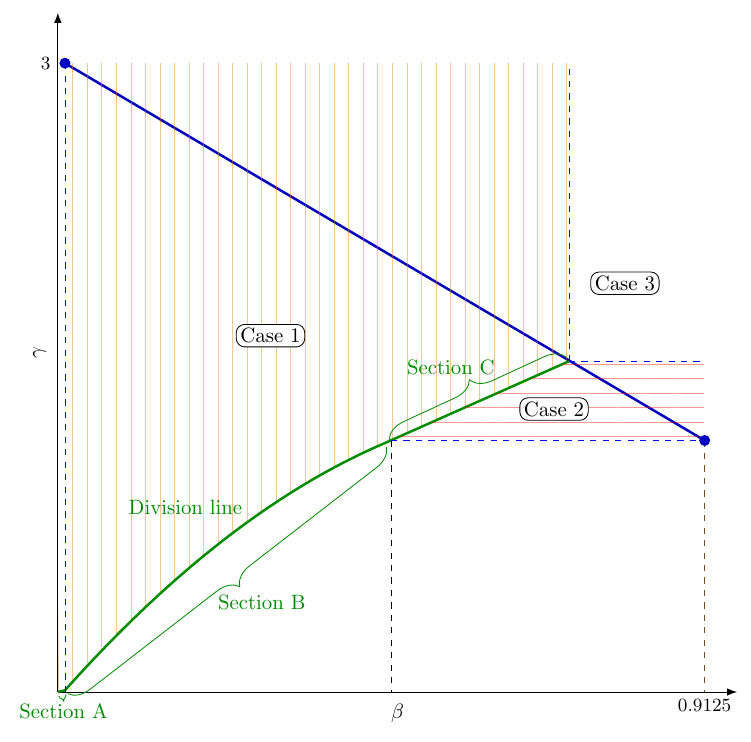}
    \caption{The division line. The marked areas correspond to parameter pairs such that the equilibrium solution remains unchanged.}\label{fig:division}
\end{figure}

We are now ready to examine whether the same queue profile can be an equilibrium for other general two-dimensional distributions of the penalty parameters.
Geometrically, this happens if the mass density of a point along the support line is shifted or spread within the intersecting ray (vertical or horizontal), shown in Figure~\ref{fig:division}. To formalize this idea, consider a general joint distribution of penalty parameters, described by the two-dimensional cumulative number of travelers,  $$\mathcal{N}(\beta_0, \gamma_0) = \left| \left\{ (\beta, \gamma) \mid \beta \leq \beta_0, \gamma \leq \gamma_0 \right\} \right|\, .$$ The same queue profile will be an equilibrium of such a general distribution if $\mathcal{N}(\beta(n(-{0.5}^{0.5} \tau_e)), \gamma(N))=0$, $\mathcal{N}(\beta(n(0)), \gamma(n(0)))=N$, and the partial derivatives along the division line do not change.
In section C the last condition can be formally stated as:
\begin{align*}
     \frac{\partial  \mathcal{N}(\beta(n(t)), \gamma(n(t')))} {\partial t} = \frac{\partial  n(t)} {\partial t} = S \quad & \forall t' \in [0, \tau_e] , t =-{0.5}^{0.5}t'\\
     \frac{\partial  \mathcal{N}(\beta(n(t)), \gamma(n(t')))} {\partial t'} = \frac{\partial  n(t')} {\partial t'} = S \quad & \forall t' \in [0, \tau_e] , t =-{0.5}^{0.5}t'.
\end{align*}
A similar formal statement can be made about section B.

A ``reverse engineering" approach of this type can be applied to any queuing profile obtained for the model defined in \ref{sec:model}, with monotonic relationship between $\beta$ and $\gamma$. Some additional attention may be needed if $\beta$ or $\gamma$ are not differentiable, but this is rather a technical concern. Note that this reverse engineering procedure does not yield a unique solution, but rather a whole class of two-dimensional distributions that share the same equilibrium.

This reverse engineering approach can be used to construct a valid equilibrium solution for cases where the connection between $\beta$ and $\gamma$ is not monotone. 
%Quote Index: R1_C10
In addition, if the distribution is ``thin" around the division line, pressure-related dynamics should be mainly local, and our stability results are likely to hold. 
However, this approach does not answer several important questions about general two-dimensional distributions, such as: a) division line identification; b) equilibrium existence; c) equilibrium uniqueness; and d) equilibrium stability. We leave these questions for future research.

Finally, it would be useful to determine whether our analysis can be extended to consider broader families of dynamics, including non-local and non-order preserving ones, or broader behavior assumptions involving equilibria with continuous choices, not necessarily related to departure time.  Exploring these directions remains a subject for future research.

\section{Concluding remarks}
This paper presents a new variant of the classic departure time choice model proposed by \cite{vickrey69}. In this variant, a single-dimensional continuous parameter represents traveler heterogeneity in earliness and lateness sensitivity. We presented a rigorous analysis of equilibrium solution existence and uniqueness. Given an instance problem with cumulative distribution functions reflecting the heterogeneity, the solution involves function integration and inversion; thus, in general, the equilibrium solution cannot be presented in closed form (except for special cases, such as linear CDFs). One key property of the equilibrium solution is that travelers depart and arrive in the \emph{a priori} order of their parameter. This result is not particularly surprising, considering that a similar sorting property has been demonstrated in variants with discrete heterogeneity (multi-class models), as well as in variants where desired arrival time is represented by a (single-dimensional) continuous parameter.  

Intuitively, we expected that the sorting property of the proposed model would simplify the identification of stable dynamics. Surprisingly, we discovered that this was not the case. In fact, we proved that the model is \emph{fundamentally unstable} in the sense that an entire family of dynamics, restricted by two plausible conditions, can not converge to the equilibrium solution. The conditions we imposed were: (a) the \emph{a priori} sorting order is maintained at all times; and (b) each traveler makes choices intending to decrease or maintain their individual cost of travel. We discussed the implications of these findings for other dynamics for the same model, as well as for other departure time choice models, including preliminary discussion of how our results can be generalized to some instances with two-dimensional heterogeneity, although many important questions remain for future study.

Part of the motivation for this work is the apparent gap between observed stability of traffic patterns in the field, and the lack of stability seen in most departure time dynamic processes that have been proposed.
Our results show that this gap remains: even when introducing and maintaining a unique departure order at equilibrium, travelers reacting to small perturbations based on local pressure at arrival time will act in a way that moves further from equilibrium.
While the stability challenge in departure time choice models remains an open question, we hope that our analysis contributes to its understanding and ultimate resolution.

\subsection*{Acknowledgements} This work was partially supported by the University Transportation Center for Understanding Future Travel Behavior and Demand  at the University of Texas (Austin), and by the Israel Science Foundation (ISF), grant no. 1361/23. 

\begin{small}
\begin{sloppypar} 
\bibliographystyle{authordate1} 

\setlength{\bibsep}{0pt}

\bibliography{references}

\end{sloppypar}
\end{small}

\appendix

\section{Alternative representations of traveler choices}
\label{sec:choicerepresentation}

Our model is defined in terms of arrival times, but it may be more natural to express a traveler's choice in terms of departure time.
The purpose of this appendix is to show how the basic behavioral choice we are modeling (departure time of an individual from a heterogeneous population) can be expressed in equivalent ways, some of which are easier to study.

The traveler choices and queueing behavior can be described in several ways.
\begin{enumerate}[label={[\arabic*]}]
\item The departure time for each traveler, $\tau(n)$.
(Here we require no ordering of departures, $n$ is an arbitrary traveler.)
\item The departure order, specified by a ``permutation'' $\phi : [0,N] \rightarrow [0,N]$.
The individual departure times [1] determine this function through $\phi(n) = \lambda \left( \left\{ n' : \tau(n') < \tau(n) \right\} \right) $ where $\lambda$ denotes the Lebesgue measure.
(As a regularity condition, we assume that $\tau(n)$ is chosen in a way that these sets are Lebesgue measurable.)
\item The cumulative departures by time $t$, $\nu_D(t)$. 
We can express this in terms of the departure times [1] as $\nu_D(t) = \lambda \left( \left\{ n : \tau(n) \leq t | \right\} \right)$. 
(We again impose as a regularity condition on $\tau(n)$ that these sets are Lebesgue measurable.)
\item The cumulative arrivals by time $t$, $\nu_A(t)$.
Cumulative arrivals can be calculated from cumulative departures [3] using the queuing dynamics, or explicitly as $\nu_A(t) = \inf_{t' < t} \{ \nu_D(t') + s(t -t ')\}$.
\item The index of the traveler arriving at time $t$, written $n(t)$, determined from the cumulative arrival profile [4] and the ordering [2]: $n(t) = \phi^{-1}(\nu_A(t))$, using the fact that the departure and arrival orders are the same. 
\item The arrival time of the traveler of index $n$, given by $t(n) = n^{-1}(t)$ as the inverse of [5].
\item The queuing delay experienced by a traveler arriving at time $t$, $Q(t) = t - \tau(n(t))$.
\item The departure time for a traveler arriving at time $t$, $D(t) = t - Q(t)$.
\end{enumerate}

In particular, [1] uniquely determines all of the other values [2--8], using the expressions given in the list above.
But it is difficult to characterize the equilibrium in terms of [1], because the queueing effects and delays are not explicit.
It is more convenient to work with the traveler arrival order [6] and the queueing delay profile [7], and we now show these variables are sufficient to fully describe the system and behavioral choices, so that the original departure times for each traveler [1] can be recovered.

\begin{proposition}
The arrival profile $t(n)$ and queueing delay profile $Q(t)$ (items [6] and [7] in the list) uniquely determine all other values in the list above.
\end{proposition}
\begin{proof}
We use implication arrows to express relationships among the above quantities; [i] $\Rightarrow$ [j] means that item [i] in the list is sufficient to determine item [i], and so forth.
The notation [i] + [j] $\Rightarrow$ [k] means that items [i] and [j] together are sufficient to determine [k].
Using this notation, we have the following:
\begin{description}
\item{[6] $\Rightarrow$ [5]:} By definition $n(t)$ and $t(n)$ are inverse functions, so knowing one is equivalent to knowing the other.
\item{[7] $\Rightarrow$ [8]:} We have $D(t) = t - Q(t)$ and $Q(t) = t - D(t)$.
\item{[6] $\Rightarrow$ [4]:} From the definition of cumulative arrivals, $\nu_A(t) = \lambda \left( \left\{ n: t(n) \leq t \right\} \right)$.
\item{[6] $\Rightarrow$ [2]:} The queue dynamics satisfy first-in, first-out, so the arrival order given by $n(t)$ is also the departure order: we can also write $\phi(n) = \nu_A(t(n))$
\item{[4] + [7] $\Rightarrow$ [3]:} Given the cumulative arrival curve and queueing delay, construct the cumulative departure curve by shifting the arrival curve to the left by the queueing delay: $\nu_D(t - Q(t)) = \nu_A(t)$.
\item{[3] + [2] $\Rightarrow$ [1]:} Apply the permutation $\phi$ to determine where a particular traveler is in the departure order, and apply the inverse cumulative departure curve: $\tau(n) = \nu_D^{-1}(\phi(n))$.
\end{description}
Conclusion: given [6] and [7], we can compute the values of all other quantities in the list, so no information is lost by analyzing the equilibrium in terms of $t(n)$ and $Q(t)$.
\end{proof}

Using this extended notation, we can reconsider the impact of order on dynamics from a different perspective. Let $k$ be a discrete index representing the dynamic evolution of the solution over steps, or `days', where $k=0$ represents the initial solution.

Consider the following corresponding notations, to refer to the state of the solution at each step:
\begin{enumerate}
\item $Q^k(t)$, queue time of a traveler arriving at $t$
\item $t^k(n)$, arrival time of traveler $n$
\item $\tau^k(n)$, departure time of traveler $n$
\item $\nu_D^k(\tau)$, cumulative quantity of departures by time $\tau$
\end{enumerate}

Suppose that $C'_{n} (t \mid Q^k)$ is positive at $t^k(n)$. Consider an individual traveler of negligible impact, i.e., a traveler who assumes that her own decisions will not modify $Q$. Assuming that such a traveler is rational, and seeking lower cost decisions, we may expect her to depart earlier, i.e., $\tau^{k+1}(n) \leq \tau^{k}(n)$. By departing earlier, the traveler may expect to arrive earlier as well, say at time $t' \leq t^k(n)$, such that $t'-Q^k(t')=\tau^{k+1}(n)$, and experience a lower cost of $C_{n} (t' \mid Q^k) \leq C_{n} (t \mid Q^k)$. If traveler $n$ had to vote about an option to depart at $\tau^{k+1}(n)$ based on this analysis, the vote would be positive.

However, other travelers may also change their departure time during the same step (iteration/day). The question we focus on is what happens if they do so in a way that does not modify their order. Under this assumption, 
$\nu_D^k(\tau^k(n))=n$,
$\nu_D^{k+1}(\tau^{k+1}(n))=n$, but since $\tau^{k+1}(n) \leq \tau^{k}(n)$ we get that the number of cumulative arrivals increase, i.e., $\nu_D^{k+1}(\tau^{k}(n)) \geq n$.

Considering a specific departure time $\tau$, we can define the flux through $\tau$ at iteration $k$ as the net number of travelers changing their departure time from being before $\tau$ to being after $\tau$, or formally, $\psi^k(\tau)=\nu_D^{k}(\tau)-\nu_D^{k+1}(\tau)$. The pressure-related dynamics suggests that when $C'_{n} (t \mid Q^k)$ is positive, the flux will be negative.

In hindsight, once the decisions of all travelers are realized, our myopic traveler may regret the decision to depart earlier. If the order remains the same, the arrival time remains the same $t^{k+1}(n)=t^{k}(n)=t$, thus scheduling costs remain the same, and the only change is a longer queueing time, since $Q^{k+1}(t)=t-\tau^{k+1}(n) \geq t-\tau^{k}(n) = Q^{k}(t)$. This outcome is perhaps counterintuitive, but it stands at the core of the concern for stability we are exploring in this study.

\section{Existence of OLP dynamics}
\label{sec:olpexistence}
This appendix discusses the existence of pressure-related mappings and provides an example of one. It is relatively intuitive to consider the derivatives as a ``direction" that can be scaled to determine the next iteration. Specifically, given a constant scaling factor, $a>0$, we may consider the mapping
\begin{equation}
\label{eqn:map1}
[Y(Q)](t) = Q(t) + a \cdot C'_{n(t)} (t \mid Q)
\,.
\end{equation}
There are two potential problems with such mappings: (a) there may be points where $Q \in \mc{Q}$ is not differentiable, meaning that continuity of the resulting profile $Y(Q)$ cannot be ensured; and (b) the conditions $[Y(Q)](t) >0$ for all $t \in \mathrm{int}~\mc{T}$, and $Q(t_2) - Q(t_1) < t_2 - t_1$ for $\tau_s \leq t_1 < t_2 < \tau_e$, are not trivial to enforce. There may be ways to restrict the state space even further, so that mappings of the form \eqn{map1} will remain within their domain, but we were not able to identify a plausible set of restrictions to do so. 

 To demonstrate that continuity of $Y(Q)$ for any $Q \in \mc{Q}$ can be obtained, we define $D(t \mid Q)$ as the distance from $t$ to the nearest non-differentiable point in $Q$. 
 Using a lower limit fraction $0<b<1$, a smoothing parameter $\alpha>0$, and a positive scaling parameter $a>0$, we consider the mapping:
\begin{align}
\label{eqn:map2}
[\tilde{Y}(Q)](t) & = \left\{ 
\begin{array}{ll}
Q(t) + \left( 1 -  e^{-\alpha \cdot D(t \mid Q ) } \right) \cdot a \cdot C'_{n(t)} (t \mid Q)  & \forall t \in \hat{\mc{T}} \\
Q(t) & \forall t \in \mc{T} \setminus \hat{\mc{T}} 
\end{array} \right.  \\
[Y(Q)](t) &= \max \left\{ b \cdot Q(t),[\tilde{Y}(Q)](t) \right\} \quad \forall t \in \mc{T}
\label{eqn:map3}
\end{align}
The lower limit fraction ensures that $Y(Q)>0$. The smoothing element $\left( 1 -  e^{-\alpha \cdot D(t \mid Q ) } \right)$ implies that if $t_0 \in \mc{T} \setminus \hat{\mc{T}}$ is a non-differentiable point, then $\lim_{t \to t_0} [\tilde{Y}(Q)](t)=Q(t)$. Therefore, by setting $[\tilde{Y}(Q)](t)=Q(t)$ continuity is maintained. Of course, \eqn{map2} and \eqn{map3} are just one example for a mapping that satisfies the conditions.

\end{document}